\documentclass[aps,epsf,nofootinbib,notitlepage,showpacs,10pt, unsortedaddress, superscriptaddress, twocolumn]{revtex4-1}

\usepackage[svgnames, table]{xcolor}

\usepackage{color}
\usepackage{float}
\usepackage{graphicx}
\usepackage{subcaption}
\usepackage{amsmath}
\usepackage{amsthm}
\usepackage{mathtools}
\usepackage{quantikz}
\newtheorem{theorem}{Theorem}
\theoremstyle{definition}

\usepackage{bm}
\usepackage{hyperref}
\usepackage[capitalise]{cleveref}

\usepackage{float}
\usepackage{multirow}
\restylefloat{table}
\usepackage{algorithm}
\usepackage{algpseudocode}
\usepackage{braket}

\usepackage{quantikz}
\theoremstyle{definition}
\newtheorem{thm}{Theorem}[section]
\newtheorem{lemma}[thm]{Lemma}
\newtheorem{definition}[thm]{Definition}

\usepackage{tikz} 
\usetikzlibrary{matrix,positioning}
\usepackage{amsfonts}
\usepackage{amssymb}

\hypersetup{
   colorlinks = true,
    urlcolor=blue
}
\usepackage[normalem]{ulem}
\usepackage{lipsum}

\graphicspath{{figures/}}
\usepackage{tikz-network}
\usetikzlibrary{fit, shapes.arrows}
\usetikzlibrary{shapes.geometric, arrows}
\usetikzlibrary{positioning,chains,fit,shapes,calc}

\usepackage{float}

\usepackage{array, booktabs, makecell, multirow}

\usepackage{amsmath,calc}

\definecolor{graphPurple}{RGB}{185,75,185}

\usepackage{adjustbox}

\begin{document}
\markboth{Sunil Vittal}
{Efficient circuits for leaf-separable state preparation}

\title{Efficient circuits for leaf-separable state preparation}

\author{Sunil Vittal}
\affiliation{
	Department of Mathematics\\ Princeton University \\ Princeton, NJ, 08544}

\author{Anthony Wilkie}
\affiliation{
	Department of Industrial and Systems Engineering\\ University of Tennessee Knoxville\\Knoxville, TN 37996}

\author{Nika Rastegari}
\affiliation{
	Department of Mathematics\\ California State University, Northridge \\ Northridge, CA, 91330}

\author{Mostafa Atallah}
\affiliation{
	Department of Industrial and Systems Engineering\\ University of Tennessee Knoxville\\Knoxville, TN 37996}
\affiliation{Department of Physics, Faculty of Science \\ Cairo University \\ Giza 12613, Egypt}

\author{Rebekah Herrman}\thanks{corresponding author}
\email{rherrma2@utk.edu}
\affiliation{
	Department of Industrial and Systems Engineering\\ University of Tennessee Knoxville\\Knoxville, TN 37996}

\begin{abstract}

Efficient state preparation is a challenging and important problem in quantum computing. In this work, we present a recursive state preparation algorithm that combines logarithmic-depth Dicke state circuits with Hamming weight encoders for efficiently preparing ``leaf-separable" quantum states. 
The algorithm is built on binary partition trees, generalized weight distribution blocks (gWDBs), and leaf-level encoders. We evaluate the performance of the algorithm by 
numerically simulating it on randomly generated target states with between 4 and 15 qubits. 
Compared to general state preparation approaches which require $O(2^n)$ CX gates, our algorithm achieves a circuit depth of $O(k\log\frac{n}{k} + 2^k)$ and uses $O(n(k+2^k))$ two-qubit gates, where $k < n$ denotes the subtree size. 
We also compare implementations of the algorithm with and without the use of ancilla qubits, providing a detailed analysis of the trade-offs in circuit depth and two-qubit gate counts. These results contribute to scalable state preparation for quantum algorithms that require structured inputs such as Dicke or near-Dicke states.

\end{abstract}

\maketitle

\section{Introduction}\label{sec: Introduction}

While quantum computers have the potential to solve classically intractable problems \cite{shaydulin2024evidence, daley2022practical, bravyi2018quantum}, creating the initial quantum state input for quantum algorithms remains a challenging task \cite{cerezo2022challenges}. General deterministic state preparation algorithms have CX gate scaling $O(2^n)$ where $n$ is the number of qubits \cite{Mottonen_2005TransformOfQuantStsUsingUnifContRot, Shende_2006SynthOfQuantLogCirc, Plesch_2011QSPWithUniversalGateDecomp}. This is prohibitive on near-term quantum devices, as controlled gates are particularly prone to error. There are general state preparation approaches that attempt to overcome this exponential scaling using ancilla qubits, computationally expensive classical preprocessing routines, quantum phase estimation, converting hard to prepare states to matrix product states, divide-and-conquer approaches, or genetic algorithms  \cite{Araujo_2021ADivAndConqAlgoForQSP, Gui_2024_Spacetime-Eff, Plesch_2011QSPWithUniversalGateDecomp, lacroix2020symmetry, aktar2022divide, kerppo2025minimizing, creevey2023gasp}. Furthermore, there are specialized sparse state preparation algorithms such as \cite{gleinig2021efficient, gonzales2025efficient, gaidai2025decomposition, de2022double, zhang2022quantum, Gui_2024_Spacetime-Eff}.

Recent work has shown that there are more efficient deterministic state preparation methods for particular quantum states such as highly symmetric states \cite{gard2020efficient, bond2025global}, GHZ states \cite{cruz2019efficient}, and Dicke states \cite{bartschi2019deterministic, bartschi2022short, PhysRevApplied.23.044014, aktar2022divide, stockton2004deterministic, yuan2025depth}, which can be prepared with logarithmic depth circuits and have applications in fields such as quantum networking \cite{prevedel2009experimental}. Further recent work has provided algorithms that prepare arbitrary Hamming-weight $k$ states, that is arbitrary quantum states that are comprised of computational basis states all of which have $k$ qubits that are ``1", with few ancilla qubits \cite{luo2025optimal, li2025preparation}.

In this work, we leverage logarithmic-depth Dicke state preparation circuits and Hamming weight encoders to develop a new algorithm that creates what we call ``leaf-separable" states, defined in Sec.~\ref{sec:algorithm}. The algorithm has $O(k\log\frac{n}{k} + 2^k)$ depth and requires $O(n(k+2^k))$ two-qubit gates where $k < n$ dictates the size of the tree needed to create the target state (see Sec.~\ref{sec: Background} for more details). The algorithm can be used to generate initial states for algorithms that require Dicke states as input or quantum states that are close to Dicke states, such as quantum algorithms for topological data analysis \cite{berry2024analyzing}, variational quantum algorithms for solving combinatorial optimization problems including portfolio optimization \cite{bartschi2020grover, brandhofer2022benchmarking, scursulim2025multiclass}, and preparing quantum many body scar states \cite{gustafson2023preparing}. 

This work is organized as follows. First, we briefly review sparse Dicke state preparation methods and Hamming weight encoders in Sec.~\ref{sec: Background}. Then, in Sec.~\ref{sec:algorithm}, we present the new leaf-separable state preparation algorithm. In Sec.~\ref{sec:experiments}, we simulate the algorithm on randomly generated states on between 4 and 15 qubits. Finally, we discuss future research directions in Sec.~\ref{sec: Discussion}.

\section{Background}\label{sec: Background}
In this section, we describe Dicke state preparation algorithms that will be used in the leaf-separable state preparation algorithm or inspired subroutines in it.

\subsection{Dicke State Preparation}\label{sec:dickestate}
The \textit{Hamming weight} of a binary string $z$ is the number of bits that have value ``1", and is denoted $h(z)$. The weight $k$ Dicke states on $n$ qubits is the uniform superposition of Hamming weight $k$ bitstrings of length $n$, that is, $\ket{D_k^n} = \binom{n}{k}^{-\frac{1}{2}}\sum_{z:h(z) = k}\ket{z}$. 
The first result from \cite{bartschi2019deterministic} proves the existence of a unitary that can map a state of the form $\ket{0^{n-\ell}1^\ell}$ to the state $\ket{D_\ell^n}$.
\begin{lemma}\cite{bartschi2019deterministic}
    There exists a unitary $U_{n,k}$ such that for all $0\leq\ell\leq k$, $U_{n,k}\ket{0^{n-\ell}1^{\ell}} = \ket{D_{\ell}^n}$, the weight $\ell$ Dicke state on $n$ qubits. This unitary can be implemented in $O(n)$ depth with $O(kn)$ $2$-qubit gates with no ancilla qubits. 
\end{lemma}
Later work indicates that one can utilize a divide and conquer approach to decrease the overhead of applying $U_{n,k}$ by distributing the initial weight $\ell$ to other qubits in the register \cite{aktar2022divide, bartschi2022short}. In this approach, one wants to create Dicke states on subsystems of $m$ and $n-m$ qubits for some natural number $m <n$. A key component of these works is the \textit{weight distribution block}, which will be used in the leaf-separable state preparation algorithm. 
\begin{definition}
    A weight distribution block is a unitary $WDB_{k}^{n,m}$ such that 
    \begin{equation}
    \resizebox{\columnwidth}{!}{$\displaystyle
        WDB_{k}^{n,m}\ket{0^{n-\ell}1^{\ell}} =\binom{n}{\ell}^{-\frac{1}{2}}\sum_{i=0}^\ell \sqrt{\binom{m}{i}\binom{n-m}{\ell-i}}\ket{0^{m-i}1^{i}}\ket{0^{n-m-(\ell-i)}1^{\ell-i}}
        $}
    \end{equation}
\end{definition}

In \cite{bartschi2022short}, the authors mention creating weight distribution trees for these blocks. Here, we provide a detailed algorithm for constructing them, \Cref{alg:partitiontree}, as the leaf-separable algorithm introduced in Sec.~\ref{sec:algorithm} heavily depends on these trees. To construct a weight distribution tree, start with $G = \lceil\frac{n}{k}\rceil$ nodes, $\lfloor\frac{n}{k}\rfloor$ of which contain $k$ qubits, and one containing $n\bmod k$ qubits. We create a binary partition tree recursively dividing the $n$ nodes into groups of $m = \lfloor\frac{G_i}{2}\rfloor$ and $n-m = \lceil\frac{G_i}{2}\rceil$ nodes, where $G_i$ is the set of leaves of vertex $i$ in the tree. We designate $G_0 = G$. Alternatively, one can construct this tree top down, starting with all qubits represented by a single node, and separating this into left and right children recursively.

\begin{algorithm}[H]
\caption{\cite{bartschi2022short} A balanced partition tree over qubits with leaf size at most $k$}
\label{alg:partitiontree}
\begin{algorithmic}
\Require Ordered list of qubits $Q$ (e.g., integers) and leaf-size threshold $k > 0$
\Ensure Root node of a binary tree where each leaf contains $\leq k$ qubits and internal nodes aggregate their descendants

\Function{BuildPartitionTree}{$Q, k$}
    \If{$|Q| \leq k$}
        \State \Return \Call{MakeNode}{$Q$} \Comment{Leaf node}
    \EndIf
    \State Partition $Q$ into contiguous chunks $C = [Q_1, Q_2, \dots, Q_m]$ where each $|Q_i| \leq k$:
    \Statex \quad $C \gets [\,Q[i : i+k] \mid i = 0, k, 2k, \dots\,]$
    \State \Return \Call{BuildFromChunks}{$C$}
\EndFunction
\\
\Function{BuildFromChunks}{$C$}
    \If{$|C| = 1$}
        \State \Return \Call{MakeNode}{$C[1]$} \Comment{Single chunk becomes a leaf}
    \EndIf
    \State $mid \gets \left\lfloor |C| / 2 \right\rfloor$
    \State $L \gets$ \Call{BuildFromChunks}{$C[1:mid]$}
    \State $R \gets$ \Call{BuildFromChunks}{$C[mid+1:|C|]$}
    \State $N \gets$ \Call{MakeNode}{$\textsc{Qubits}(L) \cup \textsc{Qubits}(R)$}
    \State \Call{SetChildren}{$N, L, R$}
    \State \Return $N$
\EndFunction
\\
\Procedure{MakeNode}{$S$}
    \State Create a node storing qubit set $S$ with null children
    \State \Return the new node
\EndProcedure
\\
\Function{Qubits}{$N$}
    \State \Return the set of qubits stored in node $N$
\EndFunction
\\
\Function{SetChildren}{$P, L, R$}
    \State Set $P$'s left child to $L$ and right child to $R$
\EndFunction
\end{algorithmic}
\end{algorithm}

From this process, we have a binary partition tree $T$ with depth $O(\log\frac{n}{k})$ and leaves of size at most $k$. For each node $v\in T$, let $Q_v$ be its qubit set with $n_v=|Q_v|$, and let $L_v$ and $R_v$ be its children. For every internal node $v$, we apply the weight distribution block $WDB_{k}^{n_v, n_{L_v}}$. At the leaves $u$ of the tree, we apply $U_{n_u,n_u}$. 

Each weight distribution block can be implemented in $O(k)$ depth and $O(k^2)$ 2-qubit gates with no ancilla qubits, so the weight distribution tree can be applied in $O(k\log\frac{n}{k})$ depth and $O(kn)$ 2-qubit gates. Noting that $n_u = O(k)$, combining the resource complexities of the weight distribution tree this with the $O(\frac{n}{k})$ applications of $U_{n_u,n_u}$, the overall state preparation algorithm can be implemented with  $O(k\log\frac{n}{k})$ depth and $O(kn)$ 2-qubit gates on all-to-all connectivity architectures \cite{bartschi2022short}. For a proof of correctness for this algorithm, we refer the readers to 
\cite{bartschi2022short}.

\subsection{Arbitrary State Preparation in the Hamming Weight $k$ Subspace}

Generalizations of Dicke state preparation algorithms to the broader Hamming weight $k$ Hilbert space require more resource intensive algorithms. Let $\mathcal{H}_{n,k} = \{\ket{\psi} = \sum_{b: h(b) = k}\alpha_b\ket{b}\}$ be the Hamming weight $k$ Hilbert space on $n$ qubits, where $\alpha_b \in \mathbb{C}$ such that $\sum_b |\alpha_b|^2 = 1$. 
Given a target state $\ket{\psi}$ and suitable initial state, we can use classical algorithms to find the circuit to prepare $\ket{\psi}\in\mathcal{H}_{n,k}$ \cite{PhysRevApplied.23.044014}. 

The algorithm contains two main subroutines (denoted S1 and S2) to identify the correct control and target qubits to iteratively mix two Hamming weight $k$ bitstrings. Given a bitstring $b$ and indices
where $b$ has ones, S1 uses Ehrlich's algorithm \cite{even1973algorithmic} to generate the next bitstring of Hamming weight $k$ for the rotation gate. An important detail is that the input string $b$ and the output string $b'$ have Hamming distance $2$, so we can isolate a rotation between the two bitstrings by identifying the positions in which they differ. Identifying these two positions gives us two qubits for which we can perform a rotation between. Given two qubits $q_1,q_2$, we apply a reconfigurable beam-splitter (RBS) gate, which acts on the subspace generated by $\{\ket{1_{q_1}0_{q_2}}, \ket{0_{q_1}1_{q_2}}\}$ via the following action
\begin{align}
    RBS(\theta,\phi)\ket{10} =&\; e^{i\frac{\phi}{2}}\cos\left(\frac{\theta}{2}\right)\ket{10} + e^{-i\frac{\phi}{2}}\sin\left(\frac{\theta}{2}\right)\ket{01}\\
     RBS(\theta,\phi)\ket{01} =&\; e^{-i\frac{\phi}{2}}\cos\left(\frac{\theta}{2}\right)\ket{10} - e^{i\frac{\phi}{2}}\sin\left(\frac{\theta}{2}\right)\ket{01} \label{eq:RBS}
\end{align}
and identity everywhere else. To maintain precision, we desire to extend this operation to act non trivially on $\ket{b}$ and $\ket{b'}$, which can be done by identifying all indices where $b$ and $b'$ share ones as control qubits. Given bitstrings $b$ and $b'$, S2 finds the suitable control qubit list and target qubit indices $q_1$ and $q_2$. 

Using the above subroutines and hyperspherical coordinates, the Hamming weight $k$ subspace preparation algorithm starts with an initial bitstring $b = 0^{n-k}1^k$ and initial state $\ket{b}$ and performs $\binom{n}{k}-1$ iterations, one for each rotation angle required to represent the target state. In the complex case, there is an extra $RZ$ gate after this process to add the correct phase on the state. The algorithm uses $O(\binom{n}{k})$ depth and $O(k \binom{n}{k})$ CNOT gates. For more information about the subroutines and algorithm, we refer the readers to \cite{PhysRevApplied.23.044014}. 
In the exposition of our algorithm, we will provide adapted versions of S2 and the core algorithm which are tailored to our use-cases.

\section{Leaf-separable state preparation algorithm}\label{sec:algorithm}

\begin{table*}[t]
\centering
\renewcommand{\arraystretch}{1.4}
\begin{tabular}{|>{\centering}p{3cm}|>{\centering}p{5cm}|>{\centering}p{4.25cm}|>{\centering}p{4.75cm}|}
\hline
\textbf{Symbol} &
\textbf{Meaning / Context} &
\textbf{Derived from previous quantities} &
\textbf{Explicit expression in terms of target amplitudes $\boldsymbol{\alpha_b}$} \tabularnewline
\hline
$\alpha_b$ &
Amplitude of computational basis state $\ket{b}$ in target state
$\displaystyle \ket{\psi} = \sum_{b: h(b)=\ell}\alpha_b\ket{b}$ &
--- &
Given directly from the target state $\ket{\psi}$ \tabularnewline
\hline
$c(I)$ &
Norm of the target state restricted to weight distribution
$I=(i_1,\dots,i_L)$, i.e., the amplitude weight of $\ket{\psi_I}$ &
$c(I) = \|P_I\ket{\psi}\| = \sqrt{\sum_{\substack{b\\I(b)=I}}|\alpha_b|^2}$ &
$\displaystyle c(I) = \sqrt{\sum_{\substack{b: I(b)=I}}|\alpha_b|^2}$ \tabularnewline
\hline
$\ket{\psi_I}$ &
Normalized projection of $\ket{\psi}$ onto the subspace with weight distribution $I$ &
$\displaystyle \ket{\psi_I} = \frac{P_I\ket{\psi}}{\|P_I\ket{\psi}\|}$ &
$\displaystyle \ket{\psi_I} = \frac{1}{c(I)}\sum_{\substack{b: I(b)=I}}\alpha_b\ket{b}$ \tabularnewline
\hline
$\alpha'_b$ &
Re-labeled amplitude normalized by reference basis element
$b^*$ (the lexicographically smallest with $\alpha_{b^*}\neq 0$ in its $I$-class) &
$\displaystyle \alpha'_b = \frac{\alpha_b}{\alpha_{b^*}}$ &
$\displaystyle \alpha'_b = \frac{\alpha_b}{\alpha_{b^*}}$ \tabularnewline
\hline
$\gamma_{u,i_u}(g_u)$ &
Local relative amplitude of leaf subsystem $u$ for weight $i_u$,
holding all other leaves fixed to their reference $g_v^*$ &
$\displaystyle \gamma_{u,i_u}(g_u) =
\frac{\alpha'_{g_1^*,\dots,g_u,\dots,g_G^*}}
{\alpha'_{g_1^*,\dots,g_u^*,\dots,g_G^*}}$ &
$\displaystyle
\gamma_{u,i_u}(g_u) =
\frac{\alpha_{g_1^*,\dots,g_u,\dots,g_G^*}}
{\alpha_{g_1^*,\dots,g_u^*,\dots,g_G^*}}$ \tabularnewline
\hline
$\eta_{u,i_u}(g_u)$ &
Normalized local amplitude used in the Hamming-weight encoder for leaf $u$ &
$\displaystyle
\eta_{u,i_u}(g_u) = 
\frac{\gamma_{u,i_u}(g_u)}
{\sqrt{\sum_{k_u}|\gamma_{u,i_u}(k_u)|^2}}$ &
$\displaystyle
\eta_{u,i_u}(g_u) = 
\sum_{g_u\in\mathcal{G}_{u,i_u}}
\frac{\frac{\alpha_{g_1^*,\dots,g_u,\dots,g_G^*}}{\alpha_{g_1^*,\dots,g_u^*,\dots,g_G^*}}
}{\sqrt{\sum_{k_u}\left|\frac{\alpha_{g_1^*,\dots,k_u,\dots,g_G^*}}{\alpha_{g_1^*,\dots,g_u^*,\dots,g_G^*}}\right|^2}}$
\tabularnewline
\hline
$\beta_i^{v,\ell}$ &
Amplitude ratio encoding the relative weight distribution $(i,\ell-i)$
across the left and right subtrees of node $v$ &
$\displaystyle \beta_i^{v,\ell} = 
\frac{\alpha_{i,\ell-i}^v}{\alpha_\ell^v}$ &
$\displaystyle
\beta_i^{v,\ell} = 
\frac{\sqrt{\sum_{\substack{b: h(b|_{Q_{L_v}})=i\\h(b|_{Q_{R_v}})=\ell-i}}|\alpha_b|^2}}
{\sqrt{\sum_{i'=0}^\ell\sum_{\substack{b: h(b|_{Q_{L_v}})=i'\\h(b|_{Q_{R_v}})=\ell-i'}}|\alpha_b|^2}}$ \tabularnewline
\hline
$\ket{\psi_{i_u}}$ &
Leaf subsystem state on $n_u$ qubits with Hamming weight $i_u$ &
$\displaystyle
\ket{\psi_{i_u}} =
\sum_{g_u\in \mathcal{G}_{u,i_u}}
\eta_{u,i_u}(g_u)\ket{g_u}$ &
{\tiny $\displaystyle
\ket{\psi_{i_u}} =
\sum_{g_u\in\mathcal{G}_{u,i_u}}
\frac{\frac{\alpha_{g_1^*,\dots,g_u,\dots,g_G^*}}{\alpha_{g_1^*,\dots,g_u^*,\dots,g_G^*}}
}{\sqrt{\sum_{k_u}\left|\frac{\alpha_{g_1^*,\dots,k_u,\dots,g_G^*}}{\alpha_{g_1^*,\dots,g_u^*,\dots,g_G^*}}\right|^2}}
\ket{g_u}$} \tabularnewline
\hline
\textbf{Leaf-separability condition} &
State $\ket{\psi}$ is leaf-separable iff each $\ket{\psi_I}$ factorizes over leaves &
$\displaystyle \alpha'_b = \prod_{u=1}^G \gamma_{u,i_u}(g_u)$ &
$\displaystyle
\frac{\alpha_b}{\alpha_{b^*}} = 
\prod_{u=1}^G
\frac{\alpha_{g_1^*,\dots,g_u,\dots,g_G^*}}
{\alpha_{g_1^*,\dots,g_u^*,\dots,g_G^*}}$ \tabularnewline
\hline
\end{tabular}
\caption{Relationships among the amplitude coefficients and subsystem factors in the leaf-separable state construction. Each quantity is shown both recursively (from prior definitions) and explicitly in terms of the original amplitudes $\alpha_b$.}
\label{tab:amplitude-hierarchy}
\end{table*}

Throughout this work, we will let $G = \lceil \frac{n}{k} \rceil$. Given a target state $\ket{\psi}\in\mathcal{H}_{n,\ell}$, define the amplitude for a basis state $\ket{b}\in\mathcal{H}_{n,\ell}$ via $\alpha_b = \braket{\psi \vert b}$. Let us arbitrarily divide the $n$ qubits into $G$ subsystems and define the weight distribution of a basis state $\ket{b}\in \mathcal{H}_{n,\ell}$ as $I(b) = (i_1,\dots, i_G)$ where $i_j$ is the Hamming weight of subsystem $j$. If $I(b)$ satisfies $\sum_j i_j = \ell$ and $0\leq i_j\leq \ell$ for all $j$, we call it \textit{valid}.
Define the projection
\begin{equation}
    P_I = \sum_{\substack{\ket{b}\in\mathcal{H}_{n,\ell}\\ I(b) = I}}\ket{b}\bra{b},
\end{equation}
which, when acting on $\ket{\psi}$, results in the state $\ket{\psi_I} = \frac{P_I\ket{\psi}}{\|P_I\ket{\psi}\|}$. Let $c(I) = \|P_I\ket{\psi}\| = \sqrt{\sum_{\substack{b\in\{0,1\}^n\\ I(b) = I}} |\alpha_b|^2}$ and $\sqrt{\sum_{I}c(I)^2} = \sqrt{\sum_{\ket{b}\in\mathcal{H}_{n,\ell}}|\alpha_b|^2} = 1$. Using this, we can represent the target state $\ket{\psi} = \sum_{I}c(I) \ket{\psi_I}$ and the target amplitudes via $\alpha_b = c(I(b))\braket{b \vert \psi_{I(b)}}$. 
\begin{definition}
    We say a state $\ket{\psi}$ is \textit{leaf-separable} if for every valid $I$, $\ket{\psi_I}\in \bigotimes_{u=1}^G \mathcal{H}_{n_u,i_u}$. That is, we can write $\ket{\psi_I} = \bigotimes_{u=1}^G \ket{\psi_{i_u}}$. 
\end{definition}

As an example of a leaf-separable state, consider a $n=4$ qubit system with $\ell = 2$, two subsystems with $n_1 = n_2 = 2$.
In this case, the subsystem basis strings of weight $0$ is $\{ 00 \}$, of weight $1$ are $\{ 01, 10 \}$, and of weight $2$ is $\{ 11 \}$.
We will define the following subsystem states
\begin{align*}
    \ket{\phi_0} &= \ket{00} \in \mathcal{H}_{2, 0}, \\
    \ket{\phi_1} &= \frac{1}{\sqrt{2}} (\ket{01} + \ket{10}) \in \mathcal{H}_{2, 1}, \\
    \ket{\phi_2} &= \ket{11} \in \mathcal{H}_{2, 2}.
\end{align*}
We claim that the following state is leaf-separable:
{\small
\begin{align}
    \ket{\psi} &= \frac{1}{\sqrt{2}} \left(\ket{\phi_1} \otimes \ket{\phi_1}\right) + \frac{1}{\sqrt{2}} \left(\ket{\phi_2} \otimes \ket{\phi_0}\right) \notag \\
    &= \frac{1}{2\sqrt{2}} \left( \ket{0101} + \ket{0110} + \ket{1001} + \ket{1010}\right) + \frac{1}{\sqrt{2}} \ket{1100} \label{eqn: leaf_state}.
\end{align}
}Notice that there are three valid weight-distributions $I = (1, 1)$, $I' = (2, 0)$, and $I'' = (0,2)$ (although $I''$ has zero support over $\ket{\psi}$).
The projectors for each weight-distribution are
{\small
\begin{align*}
    P_I &= \ket{1010}\bra{1010} + \ket{1001}\bra{1001} + \ket{0110}\bra{0110} + \ket{0101}\bra{0101}, \\
    P_{I'} &= \ket{1100}\bra{1100}, \\
    P_{I''} &= \ket{0011}\bra{0011}.
\end{align*}
}Thus,
\begin{align*}
    \ket{\psi_I} &= \ket{\phi_1} \otimes \ket{\phi_1}, & c(I) &= || P_I \ket{\psi} || = \frac{1}{\sqrt{2}}, \\
    \ket{\psi_{I'}} &= \ket{\phi_2} \otimes \ket{\phi_0}, & c(I') &= || P_{I'} \ket{\psi} || = \frac{1}{\sqrt{2}} \\
    \ket{\psi_{I''}} &= 0, & c(I'') &= 0.
\end{align*}
Here, we can see that $\ket{\psi_I} \in \mathcal{H}_{2,1} \otimes \mathcal{H}_{2,1}$ and $\ket{\psi_{I'}} \in \mathcal{H}_{2,2} \otimes \mathcal{H}_{2,0}$, thus satisfying the definition of being leaf-separable.

As an example of a state that is not leaf-separable, consider
\[
    \ket{\psi} = \frac{1}{\sqrt{2}} \left( \ket{10} \otimes \ket{01} + \ket{01} \otimes \ket{10} \right).
\]
The only valid weight-distribution with nonzero support is $I = (1, 1)$.
However $\ket{\psi_I} = \ket{\phi} \notin \mathcal{H}_{2,1} \otimes \mathcal{H}_{2,1}$.

In view of the above, any leaf-separable state may be expressed as $\ket{\psi} = \sum_{I}c(I)\bigotimes_{u=1}^G \ket{\psi_{i_u}}$ where $\ket{\psi_{i_u}}\bra{\psi_{i_u}} = Tr_{\substack{1\leq v\leq G\\v\neq u}}(\ket{\psi_I}\bra{\psi_I})$ is the state restricted to sub-system $u$.
Note that $\ket{\psi}$ need not be a product state since this is still a superposition of product states. 

Since $\alpha_b$ must separate along the leaves, we write $\ket{b} = \ket{g_1}\otimes \cdots \otimes \ket{g_G}$, meaning $c(I(b)) = c(h(g_1),\dots, h(g_G))$. We must have 
\begin{align}
    \braket{b \vert \psi_{I(b)}} =& (\bra{g_1}\otimes\cdots \otimes \bra{g_G})(\ket{\psi_{i_1}}\otimes \cdots\otimes \ket{\psi_{i_G}})\\
    =&\braket{g_1 \vert \psi_{i_1}}\cdots \braket{g_G \vert \psi_{i_G}}
\end{align}
We now define $\ket{\psi_{i_j}}$ with a basis by devising a method to compute it.

Let $g_u\in\{0,1\}^{n_u}$ with $h(g_u) = i_u$. We sort all such $g_u$ lexicographically, forming an ordered set $\mathcal{G}_{u,i_u}$. For each weight distribution $I$, choose the lexicographically smallest element $b^*\in\bigotimes_{u=1}^G \mathcal{G}_{u,i_u}$ such that $I(b^*) = I$ and $\alpha_{b^*}\neq 0$. We will write $b^* = g_1^*...g_G^*$. From here, define a re-labeling of amplitudes $\alpha_b' = \frac{\alpha_b}{\alpha_{b^*}}$. Using this, let $\gamma_{u,i_u}(g_u) = \alpha_{g_1^*\dots g_{u-1}^*g_ug_{u+1}^*\dots g_L^*}'=\frac{\alpha_{g_1^*\dots g_{u-1}^*g_ug_{u+1}^*\dots g_G^*}}{\alpha_{g_1^*\dots g_{u-1}^*g_u^*g_{u+1}^*\dots g_G^*}}$, allowing us to define
\begin{equation}
    \ket{\psi_{i_u}} = \sum_{g_u\in\mathcal{G}_{u,i_u}}\frac{\gamma_{u,i_u}(g_u)}{\sqrt{\sum_{k_u\in \mathcal{G}_{u,i_u}}|\gamma_{u,i_u}(k_u)|^2}}\ket{g_u}. \label{eq:psi_i_u}
\end{equation}
 It remains to show that for a fixed $\ket{b} = \ket{g_1}\otimes \cdots \otimes \ket{g_G}$, we can express $\alpha_b$ with these re-labeled amplitudes. \\
\begin{lemma}\label{product lemma}
    A state $\ket{\psi}$ is leaf-separable if and only if $\alpha_b' = \prod_{u=1}^G\gamma_{u,i_u}(g_u)$.
\end{lemma}
See Appendix~\ref{sec:productlemma} for the proof.

\begin{lemma}\label{amplitude formula}
    If $\ket{\psi}$ is leaf-separable, then $\alpha_b = e^{i\text{arg}(\alpha_{b^*})}c(I(b))\prod_{u=1}^G\frac{\gamma_{u,i_u}(g_u)}{\sqrt{\sum_{k_u}|\gamma_{u,i_u}(k_u)|^2}}$, where $\text{arg}(\alpha_{b^*})$ is the usual argument of a complex number.
\end{lemma}
The core idea of the proof of \Cref{amplitude formula} is to use \Cref{product lemma} in conjunction with the fact that $\prod_{u=1}^G\sqrt{\sum_{k_u}|\gamma_{u,i_u}(k_u)|^2} = \frac{c(I)}{|\alpha_{b^*}|}$. See Appendix~\ref{sec:AmplitudeProof} for details.

We want to calculate $\eta_{u,i_u}(g_u) = \frac{\gamma_{u,i_u}(g_u)}{\sqrt{\sum_{k_u}|\gamma_{u,i_u}(k_u)|^2}}$ as these will be the amplitudes for states we prepare on the leaves using the Hamming weight encoders from \cite{PhysRevApplied.23.044014}. \Cref{alg:compute-leaf-amplitudes} outlines the pseudocode for calculating $\eta_{u,i_u}(g_u) = \frac{\gamma_{u,i_u}(g_u)}{\sqrt{\sum_{g_u'}|\gamma_{u,i_u}(g_u')|^2}}$, given a target state, as described above, and we will describe how the hamming weight encoders interact to prepare leaf-separable states in \ref{sec: leaf encoders}

\begin{algorithm}[H]
\caption{Compute leaf amplitudes (``$\eta$'' values) for a target state}
\label{alg:compute-leaf-amplitudes}
\begin{algorithmic}[1]
\Require Target amplitude vector $\boldsymbol{\alpha}$ of length $N$, partition tree $T$ with leaves corresponding to disjoint qubit subsets, and list of total Hamming weights $\mathcal{K}$
\Ensure Dictionary mapping each pair (leaf node $u$, local weight $i_u$) to its Ehrlich-ordered $\eta$ amplitudes

\State Let $\mathcal{L} \gets$ leaves of $T$, sorted by the first qubit index in each leaf
\State Let $\{Q_u\}_{u\in\mathcal{L}}$ be the qubit sets of the leaves
\State Initialize empty dictionary $\mathrm{AllLeafAmps} \gets \{\}$  \Comment{stores Ehrlich-ordered etas}

\ForAll{distribution $d=(i_u)_{u\in\mathcal{L}}$ with $\sum_u i_u = \ell$}
    \ForAll{leaf $u \in \mathcal{L}$ with local weight $i_u$ from $d$}
        \If{$(u, i_u) \in \mathrm{AllLeafAmps}$}
            \State \textbf{continue} \Comment{already computed}
        \EndIf

        \State $b^* \gets$ \Call{GetReferenceState}{$d$, $\{Q_u\}$}
        \If{$b^*$ is invalid or $\alpha_{b^*} \approx 0$}
            \State \textbf{continue} \Comment{distribution not present; skip}
        \EndIf
        \State $\alpha_{b^*} \gets \boldsymbol{\alpha}[b^*]$

        \If{$i_u = 0$}
            \State $\mathrm{AllLeafAmps}[(u,i_u)] \gets [1]$
            \State \textbf{continue}
        \EndIf

        \State Initialize $\Gamma \gets \{\}$  \Comment{intermediate (standard-basis) gamma amplitudes}
        \State Let $B_u \gets$ sorted basis states on qubits $Q_u$ of weight $i_u$
        \ForAll{$g_u \in B_u$}
            \State $b' \gets$ \Call{VariantState}{$b^*, u, g_u, \{Q_u\}$}
            \State $\alpha_{b'} \gets$ $\boldsymbol{\alpha}[b']$ if $b' < |\boldsymbol{\alpha}|$ else $0$
            \State $\gamma_{g_u} \gets \alpha_{b'} / \alpha_{b^*}$
            \State $\Gamma[g_u] \gets \gamma_{g_u}$
        \EndFor

        \State $N \gets \sum_{g_u} |\Gamma[g_u]|^2$ \Comment{norm squared}
        \State $\eta_{g_u} \gets \Gamma[g_u] / \sqrt{N}$ for each $g_u$
        \State $\mathrm{AllLeafAmps}[(u,i_u)] \gets$ \Call{EhrlichAmps}{$\eta$, $Q_u$, $i_u$}
    \EndFor
\EndFor

\State \Return $\mathrm{AllLeafAmps}$
\end{algorithmic}
\end{algorithm}

In the remainder of this section, we briefly give an overview of the leaf-separable state preparation algorithm. Subsequent sections will provide proofs of the claims and properties needed in the algorithm.

\subsection{Leaf-separable algorithm overview}
For the leaf-separable state preparation algorithm, we first begin with some initial state $\ket{0^{n-\ell}1^{\ell}}$. We then create a weight distribution tree $T$ such that the leaves have size at most $k$, meaning $G = \lceil\frac{n}{k}\rceil$. Once $T$ has been created, apply $gWDB$ to all non-leaf nodes in $T$, denoted $\text{Int}(T)$, resulting in the state
\begin{equation*}
    \ket{\Psi} = \sum_{\substack{i_1,\dots, i_G\\ \sum i_u = \ell}} \prod_{w\in \text{Int}(T_v)}\frac{\alpha_{I_{L_w},I_{R_w}}^w}{\alpha_{I_w}^w}\bigotimes_{u=1}^G \ket{0^{n_u-i_u}1^{i_u}}.
\end{equation*}
Then, the application of $E^u$, defined below in Eqn.~\ref{Leaf unitary}, to each leaf node of $\ket{\Psi}$ yields the target leaf-separable state $\ket{\psi} = \sum_{b:h(b) = \ell}\alpha_b\ket{b}$. For the pseudocode, see \Cref{alg:state-prep-algorithm}.

\begin{algorithm}[H]
\caption{Overview of the leaf-separable state preparation algorithm}
\label{alg:state-prep-algorithm}
\begin{algorithmic}
\Require Target amplitude vector $\boldsymbol{\alpha}$ of length $N$, leaf size $k$.

\State Create an initial state $\ket{0^{n-\ell}1^{\ell}}$.

\State Create binary partition tree $T$ using \Cref{alg:partitiontree}

\ForAll{$v \in \text{Int}(T)$}
            \State Compute $\{\beta_i^{v,\ell}\}$ from $\{\alpha_b\}$ of target state
            \State Apply $gWDB$ to $v$
            \State \Return $\displaystyle\ket{\Psi} = \sum_{\substack{i_1,\dots, i_G\\ \sum i_u = \ell}} \prod_{w\in \text{Int}(T_v)}\frac{\alpha_{I_{L_w},I_{R_w}}^w}{\alpha_{I_w}^w}\bigotimes_{u=1}^G \ket{0^{n_u-i_u}1^{i_u}}$ \Comment{Application of $gWDB$ creates intermediate state $\ket{\Psi}$}
     \EndFor
    
    \ForAll{$u$ leaf of $\ket{\Psi}$}      \State Apply $E^u$ to $u$
\Comment{Application of $E^u$ to all $u$ creates target state}
   \EndFor

\State \Return Target state $\displaystyle\ket{\psi} = \sum_{b:h(b) = \ell}\alpha_b\ket{b}$
\end{algorithmic}
\end{algorithm}

\subsection{Generalized Weight Distribution Blocks}
Given a target state $\ket{\psi} = \sum_{b:h(b) = \ell}\alpha_b\ket{b}$, define 
\begin{align}
    \alpha_{i,j}^v =& \left\|\{\alpha_b: h(b\vert_{Q_{L_v}}) = i, h(b\vert_{Q_{R_v}}) = j\}\right\|_2\\ =& \sqrt{\sum_{\substack{b: \\h(b\vert_{Q_{L_v}}) = i,\\ h(b\vert_{Q_{R_v}})=j}}|\alpha_b|^2},
\end{align}
where $b|_{Q_{L_V}}$ is a substring containing indices $i\in Q_{L_V}$ for $b = b_1b_2\dots b_n$.
Observe $\alpha_{i,j}^v$ is the norm of $\ket{\psi}$ when restricted to amplitudes where the left side of the qubits in $v$ carries weight $i$, and the right side carries weight $j$. Also define,
\begin{equation}
    \alpha_\ell^v = \sqrt{\sum_{i=0}^\ell |\alpha_{i,\ell-i}^{v}|^2}
\end{equation}
to be the norm over all possible splits of the Hamming weight in $v$. Finally, we define
\begin{equation}\label{eq:beta}
    \beta^{v,\ell}_i = \frac{\alpha_{i,\ell-i}^v}{\alpha_{\ell}^v},
\end{equation}
which is the square root of the probability we see the weight distribution $(i,\ell-i)$ in the children of $v$. Note in \Cref{amplitude formula} that the target state carries the complex phase of the reference state $\alpha_{b^*}$. Since there is a fixed reference state, one must account for the nonzero phase $e^{i\mathrm{arg}(\alpha_{b^*})}$ in the gWDB step of the algorithm, where $\mathrm{arg}(\alpha)$ is the argument of $\alpha$. One can account for the phase by applying controlled phase gates after the controlled rotations within the gWDB gate. In view of the lack of conditioning in the tree, this is only exact when $k = \frac{n}{2}$. 

\begin{definition}
    Given a target state $\ket{\psi}$, the generalized weight distribution block at node $v$, denoted $gWDB_{k}^{n_v,m_v}$, is a unitary such that 
    \small
    \begin{equation}
       gWDB_{k}^{n_v,m_v}\ket{0^{n-\ell}1^{\ell}}  =\sum_{i=0}^\ell \beta_{i}^{v,\ell}\ket{0^{m_v-i}1^{i}}\ket{0^{n_v-m_v-(\ell-i)}1^{\ell-i}}.
    \end{equation}
\end{definition}
After computing $\beta_{i}^{v,\ell}$ for each $v$ and $i\leq \ell$, we use hyperspherical coordinates for each $i$ to determine the correct rotation angles. Recursively applying these unitaries at every node $v\in \text{Int}(T)$, where $T$ is acquired by \Cref{alg:partitiontree}, we acquire a similar intermediate state as in the Dicke State preparation algorithm.
\begin{lemma}\label{recursive gwdb lemma}
    Given a target state $\ket{\psi}$ and an initial state $\ket{0^{n-\ell}1^{\ell}}$, applying the generalized gWDB unitary recursively to a tree $T_v$ rooted at $v$ yields the state
    \begin{equation}
        \ket{\Psi} = \sum_{\substack{i_1,\dots, i_G\\ \sum i_u = \ell}} \prod_{w\in \text{Int}(T_v)}\frac{\alpha_{I_{L_w},I_{R_w}}^w}{\alpha_{I_w}^w}\bigotimes_{u=1}^G \ket{0^{n_u-i_u}1^{i_u}}
    \end{equation}
    where $I_w$ is the Hamming weight on the qubits of $w$, $I_{L_w} (I_{R_w})$ is the Hamming weight of the qubits of the left (right) child of $w$. 
\end{lemma}
The proof relies on induction on the depth of the tree. See Appendix~\ref{sec:RecursivegWDB} for details. A flow chart summarizing the gWDB block is presented in Fig.~\ref{fig:gWDB_tree_poster}.

\begin{figure}
  \centering
  \resizebox{\columnwidth}{!}{%
  \begin{tikzpicture}[
    node distance=8mm and 10mm,
    font=\footnotesize,
    io/.style        = {draw, fill=green!10, rounded corners, align=center, minimum height=7mm, inner sep=3pt, scale=1.3, font=\large},
    gWDB/.style      = {draw, fill=blue!10, rounded corners, align=center, minimum height=7mm, inner sep=3pt, scale=1.3, font=\large},
    leaf/.style      = {draw, fill=orange!10, rounded corners, align=center, minimum height=7mm, inner sep=3pt, scale=1.3, font=\large},
    arrow_label/.style = {above, sloped, midway, text width=2.2cm, align=center, yshift=1pt, font=\Large},
    thick_arrow/.style = {->, line width=0.7pt},
]

\node[gWDB] (root) {gWDB (Root $v_0$)\\$n$ qubits\\Distribute $\ell$};
\node[io, above=8mm of root] (input) {Initial State\\$\ket{0^{n-\ell}1^{\ell}}$};
\draw[thick_arrow] (input.south) -- (root.north);

\node[gWDB, below left=15mm and 12mm of root]  (L) {gWDB ($v_L$)\\$m$ qubits};
\node[gWDB, below right=15mm and 12mm of root] (R) {gWDB ($v_R$)\\$n-m$ qubits};
\draw[thick_arrow] (root.south west) -- (L.north) node[arrow_label]{$i$};
\draw[thick_arrow] (root.south east) -- (R.north) node[arrow_label]{$\ell - i$};

\node[leaf, below left=18mm and 2mm of L]  (LL) {Leaf $L_1$\\$n_1$ qubits\\Local $i_1$};
\node[leaf, below right=18mm and 2mm of L] (LR) {Leaf $L_2$\\$n_2$ qubits\\Local $i_2$};
\draw[thick_arrow] (L.south west) -- (LL.north) node[arrow_label]{$i_L$};
\draw[thick_arrow] (L.south east) -- (LR.north) node[arrow_label]{$\ell_L - i_L$};

\node[leaf, below left=18mm and 2mm of R]  (RL) {Leaf $L_3$\\$n_3$ qubits\\Local $i_3$};
\node[leaf, below right=18mm and 2mm of R] (RR) {Leaf $L_4$\\$n_4$ qubits\\Local $i_4$};
\draw[thick_arrow] (R.south west) -- (RL.north) node[arrow_label]{$i_R$};
\draw[thick_arrow] (R.south east) -- (RR.north) node[arrow_label]{$\ell_R - i_R$};

\node[io, below=25mm of $(LL)!0.5!(RR)$, minimum width=140mm] (output) {Intermediate Output\\[2pt]
    $\displaystyle\sum_{\sum i_u=\ell}\prod_{w\in\mathrm{Int}(T_v)}\frac{\alpha_{I_{L_w},I_{R_w}}^w}{\alpha_{I_w}^w}\bigotimes_{u=1}^G\ket{0^{n_u-i_u}1^{i_u}}$\quad\text{($I_w$ is the total weight on node $w$)}};

\draw[thick_arrow] (LL.south) -- (LL.south |- output.north);
\draw[thick_arrow] (LR.south) -- (LR.south |- output.north);
\draw[thick_arrow] (RL.south) -- (RL.south |- output.north);
\draw[thick_arrow] (RR.south) -- (RR.south |- output.north);

\end{tikzpicture}%
  }
  \caption{Conceptual diagram of the gWDB tree for Hamming weight distribution. An initial Hamming weight $\ell$ on $n$ qubits is recursively distributed across sub-registers down to leaf nodes.}
  \label{fig:gWDB_tree_poster}
\end{figure}
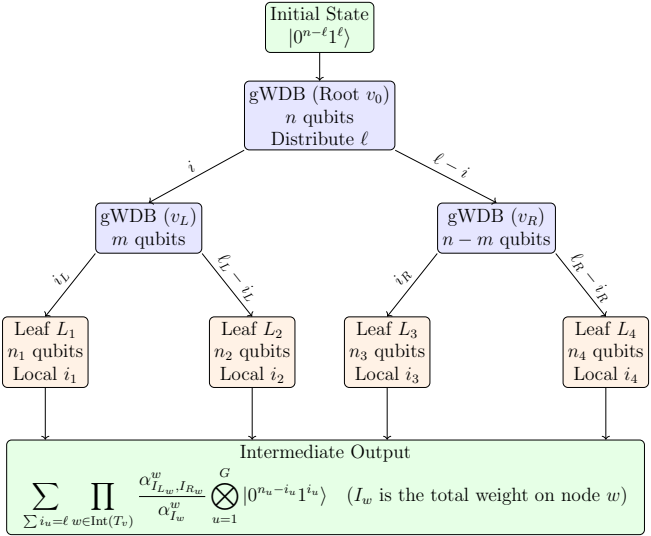

\subsection{Leaf Encoders}\label{sec: leaf encoders}
Using the appropriate HW-k encoder from \cite{PhysRevApplied.23.044014}, we can map the state $\ket{0^{n_u-i_u}1^{i_u}}\mapsto \ket{\psi_{i_u}}$ from Eq.~\eqref{eq:psi_i_u}, but this algorithm isn't designed to handle a superposition of states as an input. Define $E_{i_u}^u$, which acts on input states via
\begin{align}
    E_{i_u}^u\ket{0^{n_u-i_u}1^{i_u}} =& \sum_{g_u\in \mathcal{H}_{n_u,i_u}} \eta_{u,i_u}(g_u)\ket{g_u}\\ \eta_{u,i_u}(g_u) =& \frac{\gamma_{u,i_u}(g_u)}{\sqrt{\sum_{k_u}|\gamma_{u,i_u}(k_u)|^2}}
\end{align}
where $E_{i_u}^u$ implements the Hamming weight $i_u$ encoder on $n_u\leq k$ qubits from \cite{PhysRevApplied.23.044014} with target amplitudes $\eta_{u,i_u}(g)$ for $g\in\mathcal{H}_{u,i_u}$. We apply this unitary for all $0\leq i_u\leq \ell$, but note that $E_{i_u}^u$ acts nontrivially on $\ket{0^{n-j_u}1^{j_u}}$ when $j_u \geq i_u$. In particular, the gate decomposition of $E_{i_u}^u$ contains uniformly controlled rotations with $i_u$ controls, but this does not account for higher hamming weight states, for which these controlled gates would still act nontrivially. However, this only matters in cases where the fidelity of the algorithm is not 1. One can set $k = \lceil \frac{n}{2} \rceil$ to achieve a fidelity of 1. 

Considering that our input state to these leaf encoders is a superposition over all $0\leq i_u\leq \ell$, we need to add controls to these encoders. 
In particular, for a leaf $u$ and Hamming weight $i_u$, we would like to implement the unitary 
\begin{align}
E^{u} =& \prod_{i_u=0}^\ell \left[\ket{\phi_{i_u}}\bra{\phi_{i_u}}\otimes E_{i_u}^u + (\mathbb{I}_{n_u}-\ket{\phi_{i_u}}\bra{\phi_{i_u}})\otimes \mathbb{I}_{n_u}\right],\label{Leaf unitary}\\
\ket{\phi_{i_u}} =& \ket{0^{n_u-i_u}1^{i_u}}.
\end{align}
This leads to the result
\begin{theorem}\label{Algorithm}
    Let 
    \begin{equation*}\ket{\Psi} = \sum_{\substack{i_1,\dots, i_G \\ \sum i_u=\ell}}c(i_1,\dots, i_G)\bigotimes_{u=1}^G\left(\ket{0^{n_u-i_u}1^{i_u}}\right).
    \end{equation*}
Applying the unitary $U = \bigotimes_{u=1}^G E^u$ to $\ket{\Psi}$ yields the leaf-separable state $\ket{\psi} = \sum_{h(b) = \ell}\alpha_b\ket{b}$ provided $c(i_1,\dots, i_G)$ is obtained from a valid $I(b)$.
\end{theorem}
    See Appendix~\ref{sec:AlgorithmProof} for the proof, which relies on \Cref{amplitude formula}.

\subsection{Example of Algorithm~\ref{alg:state-prep-algorithm}}\label{sec:example}
Consider the state from Eqn.~\ref{eqn: leaf_state} which was shown to be leaf-separable:
{\small
\[
\ket{\psi_{\text{target}}} = \frac{1}{2\sqrt{2}} \left( \ket{0101} + \ket{0110} + \ket{1001} + \ket{1010}\right) + \frac{1}{\sqrt{2}} \ket{1100}.
\]
}We shall start with the initial state $\ket{0011}$ and use Algorithm~\ref{alg:state-prep-algorithm} to create $\ket{\psi_{\text{target}}}$.
First, we create a binary partition tree $T$ to partition the 4 qubits $\{0, 1, 2, 3\}$ into two sets of 2 qubits $\{0, 1\}$ and $\{2, 3\}$.
Next, we will apply the generalized weight distribution block to the root $v$, yielding:
\begin{align*}
    \ket{\Psi} &= gWDB_2^{4,2} \ket{0011} \\
    &= \sum_{i=0}^{\ell = 2} \beta_i^{v,2} \ket{0^{2-i}1^i} \ket{0^{4 - 2 - (2-i)} 1^{2-i}} \\
    &= \beta_0^{v,2} \ket{00}\ket{11} + \beta_1^{v,2} \ket{01}\ket{01} + \beta_2^{v,2} \ket{11}\ket{00} \\
    &= \sum_{\substack{i_1, i_2 \\ i_1 + i_2 = 2}} \frac{\alpha_{I_{L_v},I_{R_v}}^v}{\alpha_{I_v}^v}\bigotimes_{u=1}^{G=\lceil\frac{4}{2}\rceil = 2} \ket{0^{2-i_u} 1^{i_u}}\\
    &= \sum_{\substack{i_1, i_2 \\ i_1 + i_2 = 2}} \frac{\alpha_{i_1,i_2}^v}{\alpha_{2}^v}\bigotimes_{u=1}^{2} \ket{0^{2-i_u} 1^{i_u}},
\end{align*}
where 
\begin{align*}
    \beta_0^{v, 2} &= \frac{\alpha_{0, 2}^v}{\alpha_2^v} = 0,\;
    \beta_1^{v, 2} = \frac{\alpha_{1, 1}^v}{\alpha_2^v} = \frac{1}{\sqrt{2}},\;
    \beta_2^{v, 2} = \frac{\alpha_{2, 0}^v}{\alpha_2^v} = \frac{1}{\sqrt{2}}.
\end{align*}
More concretely, the intermediate state $\ket{\Psi}$ is
\[
    \ket{\Psi} = \frac{1}{\sqrt{2}} \ket{0101} + \frac{1}{\sqrt{2}} \ket{1100}.
\]
The gWDB angles to create $\ket{\Psi}$ from $\ket{1100}$ are
\begin{align*}
    \theta_{1,0} &= 2 \arctan\left(\frac{\sqrt{|\beta_1^{v,2}|^2 + |\beta_2^{v,2}|^2}}{|\beta_{0}^{v,2}|}\right) = \pi \\
    \theta_{1,1} &= 2 \arctan\left(\frac{\sqrt{|\beta_2^{v,2}|^2}}{|\beta_{1}^{v,2}|}\right) = \frac{\pi}{2}.
\end{align*}

As we can see, we have the correct amplitude on the $\ket{1100}$ state and we now need to evenly distribute the amplitude of $\ket{0101}$ to $\ket{0101}$, $\ket{0110}$, $\ket{1001}$ and $\ket{1010}$.
The next step of Algorithm~\ref{alg:state-prep-algorithm} is to apply $E^1$ and $E^2$ to the leaves of $\ket{\Psi}$, where:
\begin{align*}
    E^1 &= 
    \left[\ket{00}\bra{00} \otimes E_0^1 + (\mathbb{I}_2 - \ket{00}\bra{00}) \otimes \mathbb{I}_2\right] \\
    &\quad \cdot\left[\ket{01}\bra{01} \otimes E_1^1 + (\mathbb{I}_2 - \ket{01}\bra{01}) \otimes \mathbb{I}_2\right] \\
    &\quad \cdot\left[\ket{11}\bra{11} \otimes E_2^1 + (\mathbb{I}_2 - \ket{11}\bra{11}) \otimes \mathbb{I}_2\right] \\
    E^2 &= 
    \left[\ket{00}\bra{00} \otimes E_0^2 + (\mathbb{I}_2 - \ket{00}\bra{00}) \otimes \mathbb{I}_2\right] \\
    &\quad \cdot\left[\ket{01}\bra{01} \otimes E_1^2 + (\mathbb{I}_2 - \ket{01}\bra{01}) \otimes \mathbb{I}_2\right] \\
    &\quad \cdot\left[\ket{11}\bra{11} \otimes E_2^2 + (\mathbb{I}_2 - \ket{11}\bra{11}) \otimes \mathbb{I}_2\right]
\end{align*}
For $I=(1,1)$, \(\mathcal{G}_{1,1}=\{01,10\}\) and \(\mathcal{G}_{2,1}=\{01,10\}\).
Among the four nonzero strings \(\{0101,0110,1001,1010\}\), the lexicographically smallest is
\[
b^* = 0101 \quad\Rightarrow\quad g_1^*=01,\; g_2^*=01,
\]
with \(\alpha_{0101}=\tfrac{1}{2\sqrt{2}}\).
Using the target amplitudes \(\alpha_{0110}=\alpha_{1001}=\alpha_{1010}=\tfrac{1}{2\sqrt{2}}\), we get
\begin{align*}
\gamma_{1,1}(01) &= \frac{\alpha_{0101}}{\alpha_{0101}} = 1, &
\gamma_{1,1}(10) &= \frac{\alpha_{1001}}{\alpha_{0101}} = 1,\\
\gamma_{2,1}(01) &= \frac{\alpha_{0101}}{\alpha_{0101}} = 1, &
\gamma_{2,1}(10) &= \frac{\alpha_{0110}}{\alpha_{0101}} = 1.
\end{align*}
Normalizing within each leaf gives the leaf amplitudes used by the Hamming-weight encoders:
\[
\eta_{1,1}(01)=\eta_{1,1}(10)
= \frac{1}{\sqrt{2}},\quad
\eta_{2,1}(01)=\eta_{2,1}(10)
= \frac{1}{\sqrt{2}}.
\]
For $I=(2,0)$, \(\mathcal{G}_{1,2}=\{11\}\), \(\mathcal{G}_{2,0}=\{00\}\), and the only contributing basis state is
\[
b^* = 1100 \quad\Rightarrow\quad g_1^*=11,\; g_2^*=00,
\]
with amplitude \(\alpha_{1100}=\tfrac{1}{\sqrt{2}}\). Thus
\[
\gamma_{1,2}(11)=\frac{\alpha_{1100}}{\alpha_{1100}}=1,\qquad
\gamma_{2,0}(00)=\frac{\alpha_{1100}}{\alpha_{1100}}=1,
\]
so these branches are singletons and require no further leaf mixing.
In this case, $E^1_1$ and $E^2_1$ are only nontrivial unitaries, and the angles for these are:
\begin{align*}
    \tilde{\theta_1} &= 2 \arctan\left(\frac{\sqrt{|\eta_{1,1}(10)|^2}}{\eta_{1,1}(01)}\right) = \frac{\pi}{2}, \\
    \tilde{\theta_2} &= 2 \arctan\left(\frac{\sqrt{|\eta_{2,1}(10)|^2}}{\eta_{2,1}(01)}\right) = \frac{\pi}{2}.
\end{align*}

Then, by Theorem~\ref{Algorithm}, applying the unitary $E^1 \otimes E^2$ to state $\ket{\Psi}$, we have
\[
(E^1 \otimes E^2)\ket{\Psi} = \ket{\psi_{\text{target}}}.
\]

The circuit diagram for this example is in Fig.~\ref{fig:examplecircuit}.

\begin{figure} [H]
    \centering
    \resizebox{\columnwidth}{!}{%
    \input{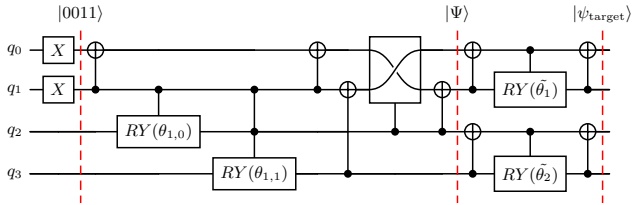}
    }
    \caption{Circuit diagram to create the target state $\ket{\psi_{\text{target}}}$ from Sec.~\ref{sec:example} using Algorithm~\ref{alg:state-prep-algorithm}.}
    \label{fig:examplecircuit}
\end{figure}

\subsection{Ancilla Optimization}
Note that for an ancilla-free method, we can uniformly control each RBS gate (Eq.~\eqref{eq:RBS}) within $E_{i_u}^u$ on the entire leaf register, meaning each RBS gate in the encoder has $\Theta(k)$ controls, but this adds a large overhead since a rotation gate with $\Theta(k)$ controls requires $\Theta(k)$ two-qubit gates to implement. Using $\Theta(G)$ ancilla qubits, however, we can eliminate this overhead and add only one extra control per RBS gate. The details of this optimization are below.

Recall the intermediate state $\ket{\Psi}$. Note that for leaf-separable states, while the overall system is entangled, the individual leaf registers separate, meaning $\ket{\psi_I}$ is a rank one vector in the subspace $\mathcal{H}_{n,\ell} = \bigotimes_{u=1}^G \mathcal{H}_{n_u,i_u}$. As a result, when applying the leaf encoders we can apply $\bigotimes_{u=1}^G E_{j}^u$ for every $0\leq j\leq \ell$ without regards for the order the $i_u$ show up in. This is implied in \cref{Leaf unitary} since $E^u$ is defined without care for the order the $i_u$ occur, with respect to some Gray code ordering, that is, an ordering of bitstrings such that two successive bitstrings only differ in a single bit. For example, $00, 01, 11, 10$ is a Gray code on length two bitstrings.

Due to this property, we apply the Hamming weight encoders $E_{j}^u$ in increasing order. Since each encoder requires a starting state of $\ket{0^{n-j}1^j}$, we deterministically know which state we require. As a result, for each leaf, we perform a $\ket{0}$-controlled CNOT gate controlled by the $(n-j)^{th}$ qubit in the leaf register onto an ancilla qubit. In every RBS gate for $E_{j}^u$, the ancilla qubit acts as an additional control, giving us the $\ket{1}$-controlled unitary $CE_{j}^u$. 

To see that this works, note that for $a < j$, CNOTs from previous $CE_a^u$ will ensure all states $\ket{\psi}\in\mathcal{H}_{n_u, a}$ are of the form $\ket{1}\otimes \ket{\psi}$ on the overall subsystem including the ancillary qubit. Moreover, $E_j^u$ acts trivially on $\mathcal{H}_{n_u,a}$, $CE_{j}^u$ acts trivially on all such $\ket{\psi}$. 

For $b>j$, we have not yet acted on the starting state $\ket{0^{n-b}1^b}$, but recall $E_{j}^u$ acts nontrivially on $\ket{0^{n-b}1^b}$. The overall state on this subsystem is $\ket{0}\otimes \ket{0^{n-b}1^b}$ since all previous $\ket{0}-$controlled CNOTs act trivially on $\ket{0^{n-b}1^b}$ so the extra control ensures $CE_j^u$ acts trivially on $\mathcal{H}_{n_u,b}$.

Since we use a single ancilla for each leaf $u$, we require $G = \lceil\frac{n}{k}\rceil$ ancilla qubits. Moreover, recall every ancilla-free leaf encoder $E_{i_u}^u$ contains $\binom{n_u}{i_u}-1$ $(k-1)-$controlled RBS gates. With a single ancilla qubit per leaf, the number of controls reduce from $\Theta(k)$ to $\Theta(i_u)$ per RBS gate, meaning $CE_{i_u}^u$ costs $\Theta(\binom{k}{i_u})$ gates, which is an improvement over the $\Theta(k\binom{n_u}{i_u})$ gates in $E_{i_u}^u$. In \Cref{sec:experiments}, we will see the resource differences between the ancilla and uniformly-controlled methods.

\subsection{When is the gWDB tree exact?}
In general, one notices that \Cref{recursive gwdb lemma} describes a state with coefficients $\prod_{w\in \text{Int}(T_v)}\frac{\alpha_{I_{L_w},I_{R_w}}^w}{\alpha_{I_w}^w}$ and these coefficients are not necessarily $c(i_1,\dots, i_G)$ as described for arbitrary leaf-separable states. Without ancilla qubits, the gWDB tree produces an approximation of a target state since in each branch of the tree, the weight on a node is not conditioned on the knowledge of the weight on the node's sibling(s), resulting in a loss of the mutual information of the target probability distribution.

The gWDB tree is an example of a Chow-Liu Tree \cite{1054142}, so there are known bounds on the information loss. It is not clear whether the output state $\ket{\Psi}$ is optimal since we cannot verify that the KL divergence between the target intermediate state and $\ket{\Psi}$ is minimal. In general, our state preparation is exact when $k = \frac{n}{2}$ for both real and complex target states and also $k = 1$ for real target states since there is nothing to condition on for these cases.

If one wants to exactly prepare $\ket{\Psi}$, notice by using stars-and-bars counting, the number of valid weight distributions for some $n,k,\ell$ is less than $\binom{n}{\ell}$, so one can apply the sparse fixed hamming weight algorithm from \cite{PhysRevApplied.23.044014} for ancilla-free state-prep. Future work can also involve finding an ancilla-based method to tighten this approximation--perhaps using the $\Theta(G)$ ancilla qubits before optimizing the leaf encoder step--or make this step deterministic. We showcase numerical simulations of the fidelity loss between the approximation $\ket{\Psi}$ and the target intermediate state in \Cref{sec:experiments}.

\subsection{Resource complexities and extensions to mixed Hamming weight states}
\begin{theorem}\label{Algorithm resource complexity}
\Cref{alg:state-prep-algorithm} can be implemented with no ancilla qubits and $O(k(\log\frac{n}{k} + 2^k))$ depth and $O(nk2^k)$ two-qubit gates. With $G = \lceil\frac{n}{k}\rceil$ ancillas, we can implement the algorithm with $O(k\log\frac{n}{k} + 2^k)$ depth and $O(n(k+ 2^k))$ two-qubit gates.
\end{theorem}
\begin{proof}
    Using depth calculations from \cite{bartschi2022short}, we note that $gWDB_{k}^{n,m}$ can be implemented in $O(k)$ depth and $O(k^2)$ two-qubit gates. From here, note that on the leaves, we apply $O(\ell)$ Hamming weight encoders on $k$ qubits. Since $\ell\leq k$, in the worst case $\ell = k$, we are performing arbitrary state preparation on the leaves. Each weight $i$ encoder has $O(\binom{k}{i})$ RBS gates, meaning our leaf encoders have depth $k\sum_{i=0}^\ell \binom{k}{i}$, where the extra $k$ comes with the extra controls. When $\ell = k$, this sum is $k2^k$, so the overall depth when ignoring the choice of $\ell$ and focusing on worst case bounds is $O(k(\log\frac{n}{k} + 2^k))$. 
    
    When looking at gate costs, \cite{PhysRevApplied.23.044014} states that arbitrary state preparation on $k$ qubits using the Hamming weight encoders costs $O(k2^k)$ CNOTs, and since we use $\Theta(k)$ extra two-qubit gates for the uniformly controlled rotations, we use $O(k^22^k)$ two-qubit gates per leaf. Since we have $O(\frac{n}{k})$ leaves, this step uses $O(nk2^k)$ two-qubit gates. Since the gWDB tree can be implemented in $O(nk)$ two qubit gates, overall, we can implement this algorithm using $O(nk2^k))$ two-qubit gates and no ancilla qubits.

    When using $G$ ancillas, we eliminate the need of using $\Theta(k)$ controls per rotation gate, and use $O(1)$ extra controls per rotation gate, eliminating a factor of $k$ from the above calculations since we've reduced to arbitrary state preparation from \cite{PhysRevApplied.23.044014}. That is, with $G$ ancillas, the algorithm has depth $O(k\log\frac{n}{k} + 2^k)$ and requires $O(n(k+ 2^k))$ two-qubit gates.
\end{proof}

To extend our algorithm to mixed Hamming weight states $\ket{\psi}\in \bigoplus_{\ell=0}^k \mathcal{H}_{n,\ell}$, the adaptation here is to simply use the symmetric state preparation circuit from \cite{bartschi2019deterministic}, which prepares the state $\sum_{\ell=0}^n \alpha_\ell\ket{0^{n-\ell}1^{\ell}}$ in $O(n)$ depth and $O(n)$ gates. Since the gWDB tree works for Hamming weights $k\leq \lfloor\frac{n}{2}\rfloor$, we cannot use a general state, but instead can prepare states $\ket{\psi}\in\bigoplus_{\ell=0}^{\lfloor\frac{n}{2}\rfloor} \mathcal{H}_{n,\ell}\subset(\mathbb{C}^2)^{\otimes n}$, which has dimension $O(2^{\frac{n}{2}})$.

Defining $\alpha_\ell = \sqrt{\sum_{i=0}^\ell |\alpha_{i,\ell-i}|^2}$, we see that this neatly integrates into the gWDB tree. At the root $r$, notice that when working with a single Hamming weight $\beta_{i}^{r,\ell} = \alpha_{i,\ell-i}^r$, and with the use of multiple Hamming weights we will have a nontrivial denominator $\alpha_{\ell}^r$. But since the input state to the gWDB tree will be $\sum_{\ell=0}^k \alpha^r_\ell\ket{0^{n-\ell}1^{\ell}}$, the $\alpha_\ell^r$ will cancel in the first layer of the tree. As a result, the rest of the tree continues the same, and we now have the following corollary:\\

\textbf{Corollary}: Given an input state $\sum_{\ell=0}^k \alpha_{\ell}^v\ket{0^{n-\ell}1^\ell}$ such that the target amplitudes in $\mathcal{H}_{n,\ell}$ separate across leaves for all $\ell\leq k$, the leaf-separable state preparation algorithm prepares a state in the subspace $\bigoplus_{\ell=0}^k \mathcal{H}_{n,\ell}\subset(\mathbb{C}^2)^{\otimes n}$. This state can be prepared with $O(k\log\frac{n}{k} + 2^k + k)$ depth, $O(nk+n2^k + k)$ two-qubit gates, and $O(\frac{n}{k})$ ancilla qubits.


\section{Numerical simulations}\label{sec:experiments}

In this section, we numerically simulate the leaf-separable state preparation algorithm to show resource usage and fidelity loss for arbitrary leaf-separable states. In the resource usages, we assume $k = \lceil\frac{n}{2}\rceil$ to simulate worst case asymptotics and to measure differences between exact state preparation algorithms. For the fidelity plots, we randomly generate 200 leaf-separable states with between 4 and 15 qubits and calculate the fidelity using the ancilla-free version of the algorithm, which has lower fidelity compared to using ancilla in the algorithm. 

\cref{fig:two-qubit-gates}, and \cref{fig:total-gates} display the two-qubit and total gate costs across our algorithm with and without ancillas, the arbitrary HW-K encoder \cite{PhysRevApplied.23.044014}, and arbitrary M\"{o}tt\"{o}nen state preparation \cite{Mottonen_2005TransformOfQuantStsUsingUnifContRot}, respectively. Note that while the leaf-separable state preparation algorithm requires more gates than the HW-k encoder and M\"{o}tt\"{o}nen approaches to prepare states on up to 13 qubits, note that there is a crossing point at 14 qubits. This indicates that the leaf-separable algorithm may be able to more efficiently prepare leaf-separable states with a high number of qubits. However, we were unable to simulate larger state preparation circuits due to limitations on computational resources. 
\begin{figure}
    \centering
    \includegraphics[width=\linewidth]{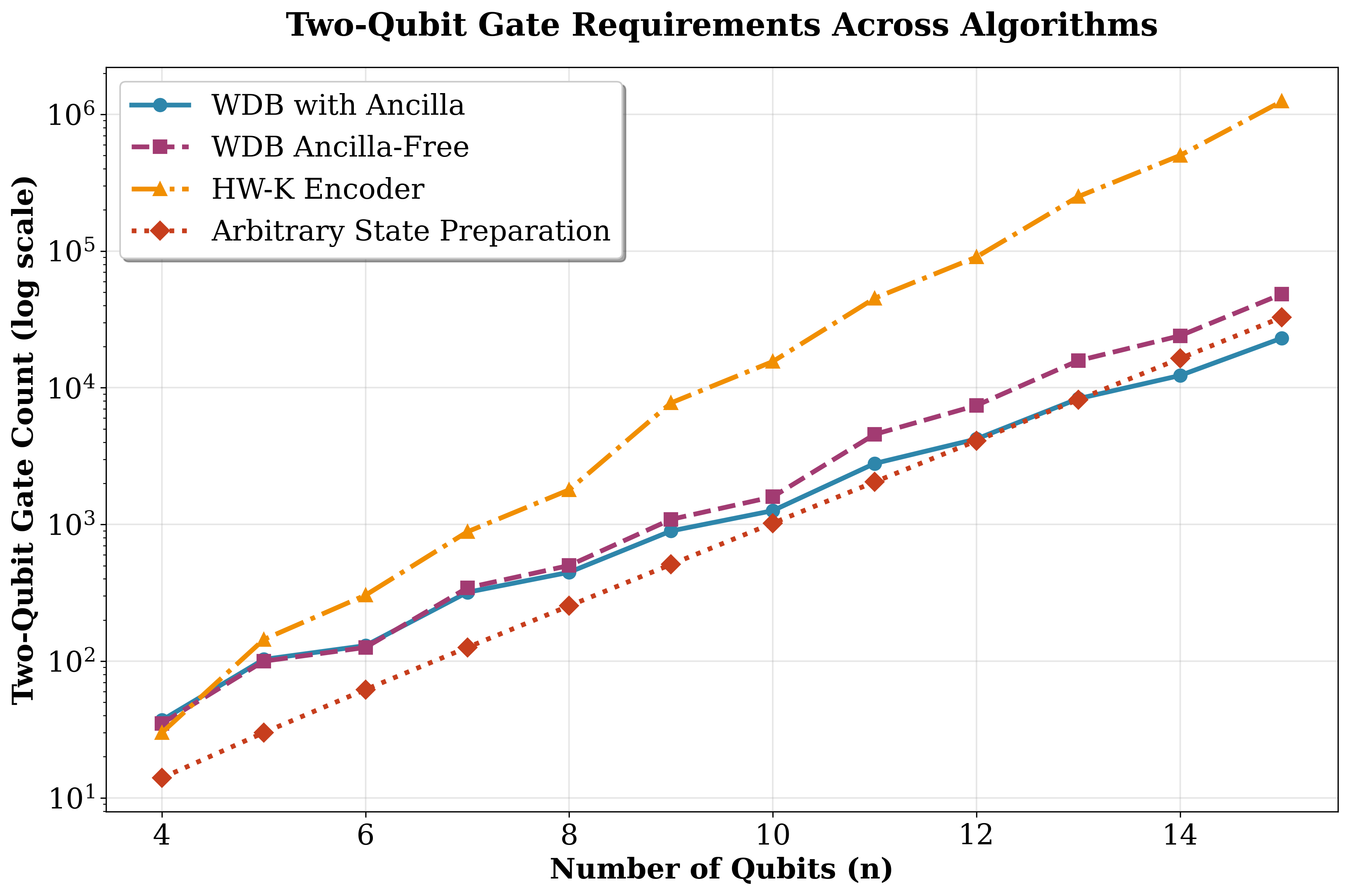}
    \caption{The 2-qubit gate count required to implement a HW-k encoder in the worst case when $k = \lceil \frac{n}{2} \rceil$ for $n$ between 4 and 15 qubits. The WDB with Ancilla line shows our algorithm performance with $2$ ancilla qubits.}
    \label{fig:two-qubit-gates}
\end{figure}
\begin{figure}
    \centering
    \includegraphics[width=\linewidth]{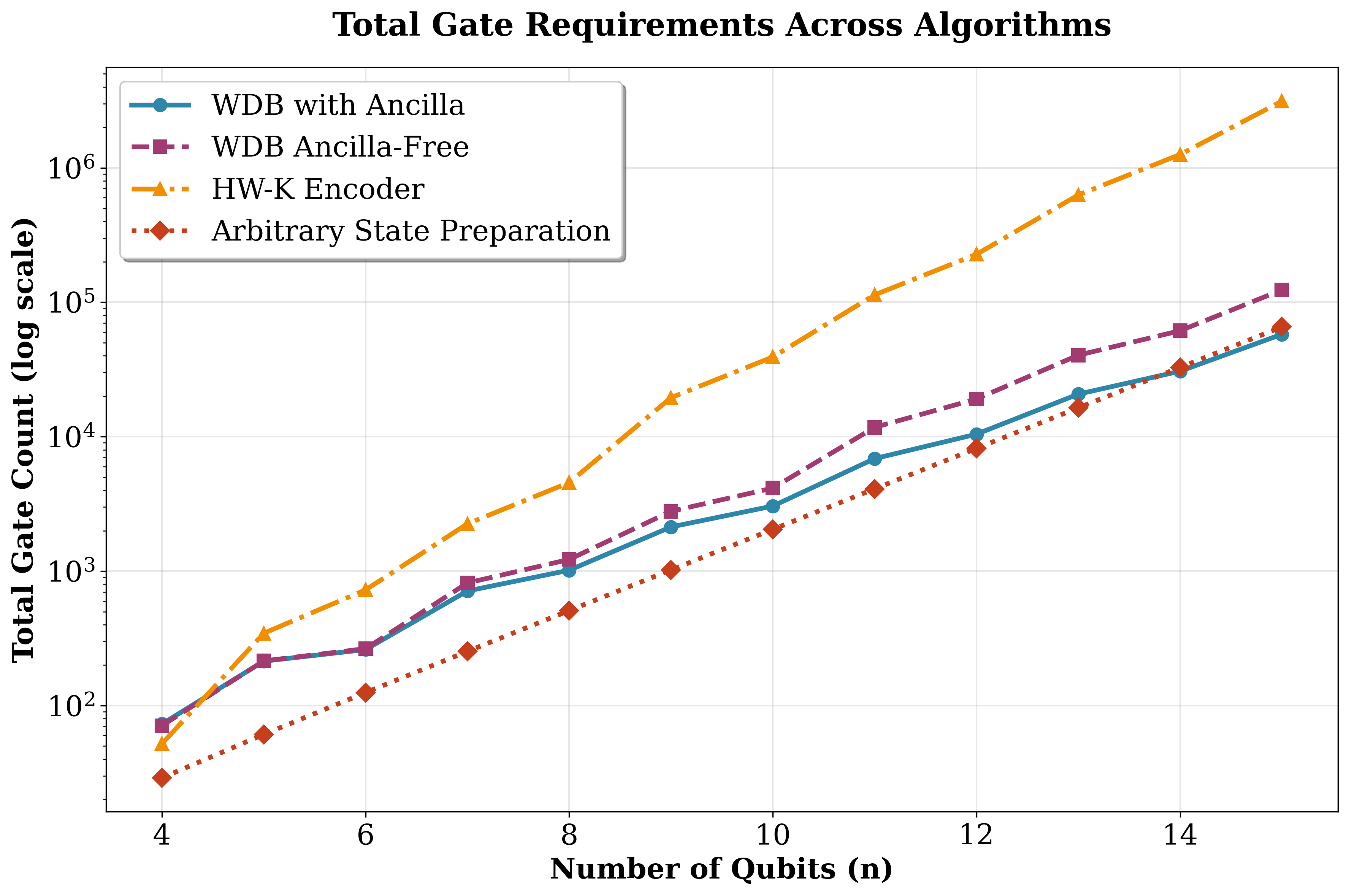}
    \caption{The total gate count required to implement a HW-k encoder in the worst case when $k = \lceil \frac{n}{2} \rceil$ for $n$ between 4 and 15 qubits.}
    \label{fig:total-gates}
\end{figure}

\begin{figure}
    \centering
    \includegraphics[width=\linewidth]{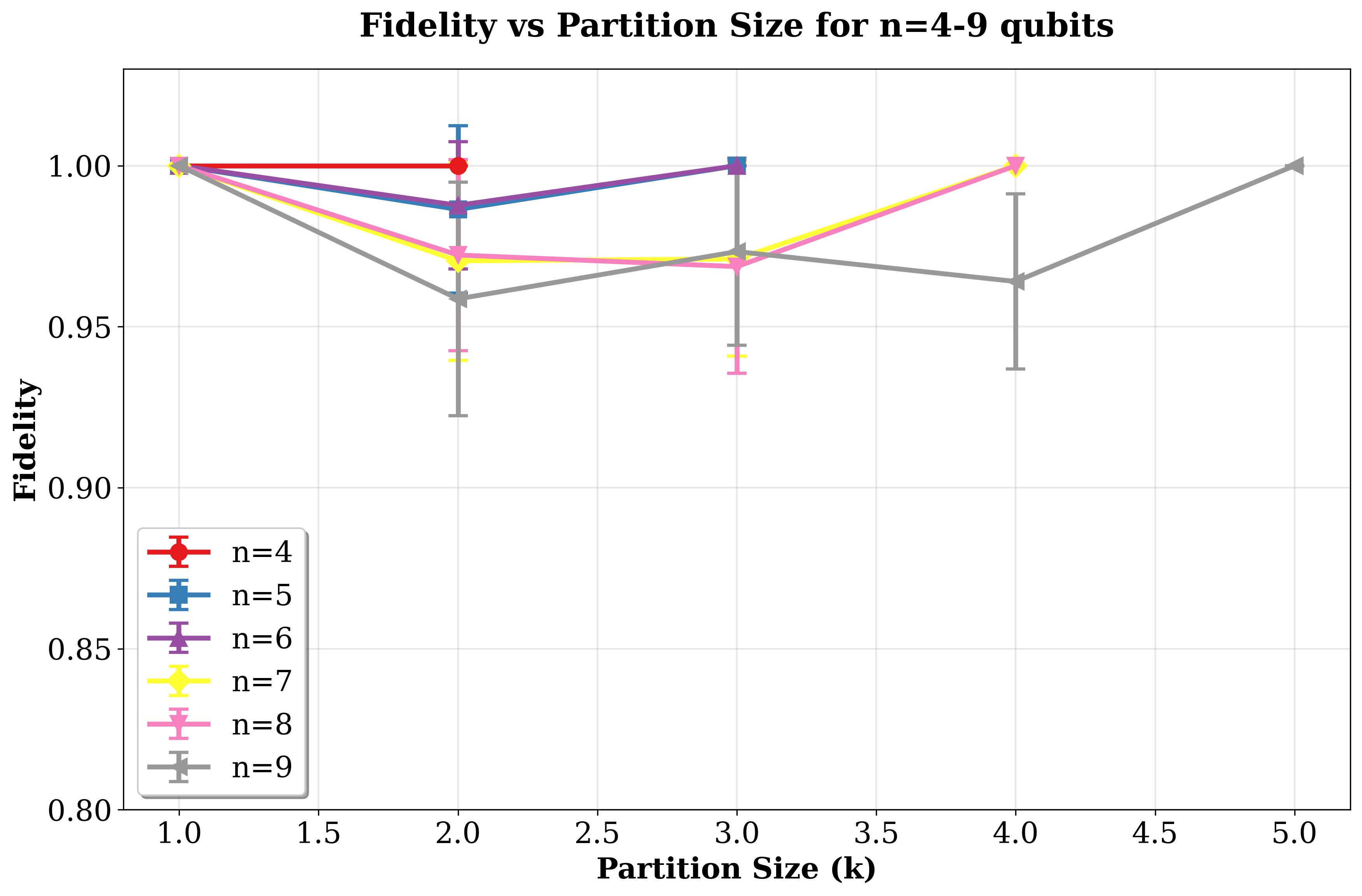}
    \caption{The average fidelity when using the leaf-separable state preparation algorithm to prepare 200 random states with real amplitudes on between 4 and 9 qubits.}
    \label{fig:fidelity-4-9}
\end{figure}

\begin{figure}
    \centering
    \includegraphics[width=\linewidth]{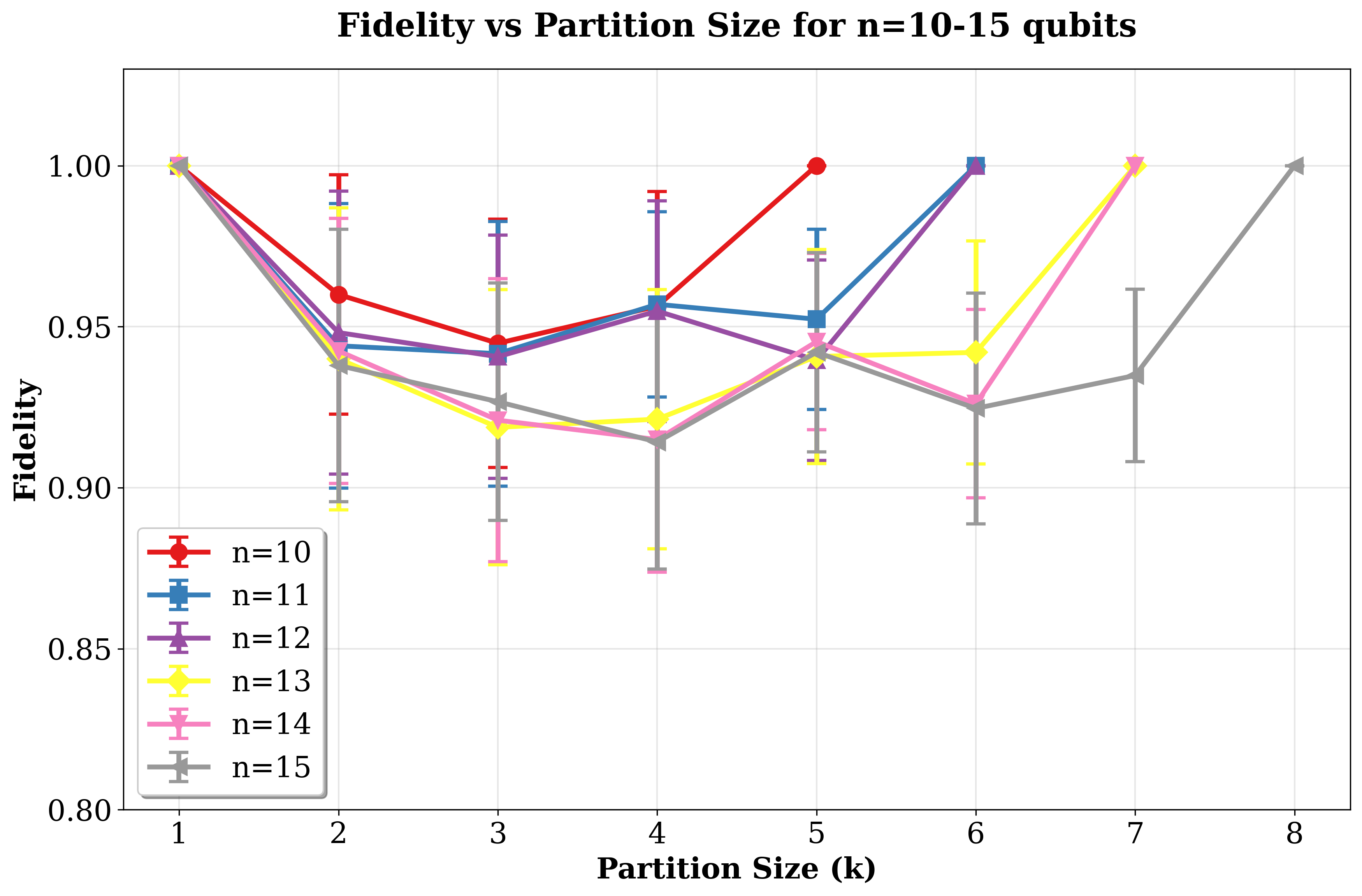}
    \caption{The average fidelity when using the leaf-separable state preparation algorithm to prepare 200 random states with real amplitudes on between 10 and 15 qubits.}
    \label{fig:fidelity-10-15}
\end{figure}

In Figs.~\ref{fig:fidelity-4-9}-\ref{fig:fidelity-10-15} we show the fidelity of the leaf-separable state preparation algorithm on states with between 4 and 15 qubits as a function of $k$. Note that the fidelity for preparing all states is, on average, above .9, and that the fidelity depends on both $n$ and $k$, with small and large values of $k$ resulting in higher fidelity, as expected. This increase in fidelity for large $k$, can in part, be attributed to the fact that the leaf-separable state preparation algorithm requires more gates as $k$ increases, as seen in Figs.~\ref{fig:total-gates}. It is also worth noting that the ancilla-free variation of the algorithm was used to generate these plots, as it has the lowest fidelity, so using ancilla in the algorithm will result in an even higher average fidelity.



\section{Discussion}\label{sec: Discussion}

We demonstrated an efficient leaf-separable state preparation algorithm for preparing initial quantum states that are either Dicke states or close approximations thereof. Our method leverages logarithmic-depth Dicke state preparation circuits combined with Hamming weight encoders to generate what we call leaf-separable states. Our results indicate that this approach accurately and efficiently prepares symmetric Dicke states, and generalizes to mixed Hamming weight superpositions by adapting the symmetric state preparation circuit and Hamming weight $k$ encoder circuits \cite{bartschi2019deterministic, PhysRevApplied.23.044014}. In contrast to general state preparation approaches, which require $O(2^n)$ CX gates, our construction achieves a circuit depth of $O(k\log\frac{n}{k} + 2^k)$ depth and uses $O(n(k+2^k))$ two-qubit gates, where $k < n$ denotes the size of the encoding tree required for the target state. 

However, a primary limitation of the approach is that it relies on the probability of measuring the weight distribution, $I(b)$ of each state $\ket{b}$. Thus, the approach is not deterministic and suffers from information loss. Exactly modeling this information loss as a function of the partition tree $T$ is required to find a regime in which the leaf-separable state preparation algorithm can be used with high fidelity. Developing a deterministic equivalent of the algorithm would circumvent the need for calculating the information loss. Other future directions of research include how to map a leaf-separable state to an arbitrary quantum state that has a similar number of nonzero amplitude basis states as the leaf-separable state and testing the leaf-separable state preparation algorithm on hardware for larger states. 


\section*{Author contributions}
SV developed the leaf-separable algorithm and wrote code. AW developed examples, checked the logical consistency of the definitions and proofs, and generated data for plots. NR developed examples and wrote code. MA developed the GitHub repository and organized and cleaned the code. RH developed the initial project idea and acquired funding.

All authors wrote, read, edited, and approved the final manuscript.

\section*{Acknowledgments}
N.R., S.V., and R.H. acknowledge NSF CNS-2244512. R.H. and M.A. acknowledge DE-SC0024290. R. H. and A. W. acknowledge NSF CCF 2210063.

\section*{Competing interests}
All authors declare no financial or non-financial competing interests. 

\section*{Code and Data Availability}
\label{sec:codeAndDataAvail}
The code and data for this research can be found at \url{https://github.com/Mostafa-Atallah2020/state-prep-reu2025}.

\appendix
\onecolumngrid
\newpage

\section{Proof of \Cref{product lemma}}\label{sec:productlemma}

  Suppose that for any valid weight distribution $I = (i_1,\dots, i_G)$, we have $\ket{\psi_I} = \bigotimes_{u=1}^G \ket{\psi_{i_u}}$. As a result, for a state $\ket{b} = \ket{g_1}\otimes \cdots \otimes \ket{g_G}$ with $I(b) = I$, we have $\alpha_b = c(I)\braket{g_1\vert \psi_{i_1}}\cdots \braket{g_G\vert \psi_{i_G}}$. Now, recalling $\alpha_b' = \frac{\alpha_b}{\alpha_{b^*}}$, we can use the splitting of $\alpha_b$ to see
    \begin{equation*}
        \frac{\alpha_b}{\alpha_{b^*}} = \frac{\braket{g_1\vert \psi_{i_1}}\cdots \braket{g_G\vert \psi_{i_G}}}{\braket{g_1^*\vert \psi_{i_1}}\cdots \braket{g_G^*\vert \psi_{i_G}}}
    \end{equation*}
    Recall that $\gamma_{u,i_u}(g_u) = \alpha_{g_1^*\cdots g_{u-1}^*g_ug_{u+1}^*\cdots g_G^*}'$. Using Eq.~\eqref{eq:psi_i_u}, we have
    \begin{equation*}
        \gamma_{u,i_u}(g_u) = \frac{\braket{g_1^*\vert \psi_{i_1}}\cdots \braket{g_{u-1}^*\vert \psi_{i_{u-1}}}\braket{g_u\vert \psi_{i_u}}\braket{g_{u+1}^*\vert \psi_{i_{u+1}}}\cdots\braket{g_G^*\vert \psi_{i_G}}}{\braket{g_1^*\vert \psi_{i_1}}\cdots \braket{g_G^*\vert \psi_{i_G}}} = \frac{\braket{g_u\vert \psi_{i_u}}}{\braket{g_u^*\vert \psi_{i_u}}}
    \end{equation*}
    As a result, we have $\prod_{u=1}^G\gamma_{u,i_u}(g_u) = \prod_{u=1}^G \frac{\braket{g_u\vert \psi_{i_u}}}{\braket{g_u^*\vert \psi_{i_u}}} = \alpha_b'$.

    Conversely, suppose $\alpha_b' = \prod_{u=1}^G\gamma_{u,i_u}(g_u)$, so that $\alpha_b = \alpha_{b^*}\prod_{u=1}^G\gamma_{u,i_u}(g_u)$. Note that we can represent $\ket{\psi_I} = \sum_{b: I(b) = I}\frac{\alpha_b}{c(I)}\ket{b}$. Breaking this down further, we have
    \begin{align*}
        \sum_{b: I(b) = I}\frac{\alpha_b}{c(I)}\ket{b} =& \frac{1}{c(I)}\sum_{g_1\in\mathcal{G}_{1, i_1}}\cdots \sum_{g_G\in\mathcal{G}_{G,i_G}}\alpha_{g_1\cdots g_G}\ket{g_1}\otimes \cdots \otimes \ket{g_G} \\
        =& \frac{1}{c(I)}\sum_{g_1\in\mathcal{G}_{1, i_1}}\cdots \sum_{g_L\in\mathcal{G}_{G,i_G}}\alpha_{b^*}\prod_{u=1}^G \gamma_{u,i_u}(g_u)\ket{g_1}\otimes \cdots \otimes \ket{g_G} \\
        =& \frac{\alpha_{b^*}}{c(I)}\bigotimes_{u=1}^G\left(\sum_{g_u\in\mathcal{G}_{u,i_u}}\gamma_{u,i_u}(g_u)\ket{g_u}\right) \\
        =& \frac{\alpha_{b^*}c'}{c(I)}\bigotimes_{u=1}^G\ket{\psi_{i_u}} 
    \end{align*}
    where $c'$ is the normalization constant from Eq.~\eqref{eq:psi_i_u}, giving a leaf factorization for each $\ket{\psi_I}$. \qed

\section{Proof of \Cref{amplitude formula}}\label{sec:AmplitudeProof}

Since $\ket{\psi_I} = \bigotimes_{i\in I}\ket{\psi_i}$, we know $\alpha_b' = \prod_{u=1}^G\gamma_{u,i_u}(g_u)$. As a result,
    \begin{equation}\label{eq: Lem III.3 expanded eq}
        e^{i\text{arg}(\alpha_{b^*})}c(I(b))\prod_{u=1}^G\frac{\gamma_{u,i_u}(g_u)}{\sqrt{\sum_{g_u'}|\gamma_{u,i_u}(g_u')|^2}} = c(I(b))\frac{\alpha_b}{|\alpha_{b^*}|}\prod_{u=1}^G\frac{1}{\sqrt{\sum_{g_u'}|\gamma_{u,i_u}(g_u')|^2}}
    \end{equation}
    Focusing on the product, note
    \begin{equation*}
        \prod_{u=1}^G\sqrt{\sum_{g_u'}|\gamma_{u,i_u}(g_u')|^2} = \sqrt{\sum_{g_1'}\cdots\sum_{g_G'}\prod_{u=1}^G|\gamma_{u,i_u}(g'_u)|^2} = \sqrt{\sum_{b': I(b') =I(b)}|\alpha_{b'}'|^2} = \frac{1}{|\alpha_{b^*}|}c(I(b))
    \end{equation*}
    Substituting this into Eq.~\eqref{eq: Lem III.3 expanded eq}, we have 
    \begin{equation*}
        c(I(b))\frac{\alpha_b}{|\alpha_{b^*}|}\prod_{u=1}^G\frac{1}{\sqrt{\sum_{g_u'}|\gamma_{u,i_u}(g_u')|^2}} = \alpha_b
    \end{equation*}
    as desired. \qed

\section{Proof of \Cref{recursive gwdb lemma}}\label{sec:RecursivegWDB}
We prove this via induction on the depth of the tree. For the base case of the tree having zero depth, the only node is a leaf and no gWDB unitaries are applied. As a result, the only choice of $i$ is $i=\ell$, so $c(\ell) = 1$ and $n_u = n$. Moreover, since no unitaries are applied, the output state should be the input state $\ket{0^{n-\ell} 1^\ell}$, which matches with $\ket{\Psi}$ for this case.
    Now, suppose that for a tree rooted at $v$, the claim holds for its children $L_v, R_v$. Applying $gWDB_{k}^{n_v,m_v}$ on an input $\ket{0^{n-\ell} 1^\ell}$ at $v$, we have
    \begin{equation*}
       gWDB_{k}^{n_v, m_v}\ket{0^{n-\ell} 1^\ell} = \sum_{i=0}^\ell \beta_{i}^{v,\ell}\ket{0^{m_v-i}1^{ i}}\ket{0^{n_v-m_v+i-\ell}1^{\ell-i}}. 
    \end{equation*}
    Noting that $L_v$ is the $m_v$ qubit sub-register and $R_v$ is the $n_v-m_v$ qubit sub-register, the induction hypothesis indicates that after completing the gWDB tree, we have the output state
    \begin{align*}
        \ket{\Psi} = \tiny{\sum_{i=0}^\ell \beta_{i}^{v,\ell}\left(\sum_{\substack{r_1,\dots, r_a,\\ \sum r_u = i}}\prod_{w\in \text{Int}(T_{L_v})}\frac{\alpha_{I_{L_w},I_{R_w}}^w}{\alpha_{I_w}^w}\bigotimes_{u=1}^a\ket{0^{n_u-r_u}1^{r_u}}\right) \left(\sum_{\substack{s_{a+1},\dots, s_G,\\ \sum s_u = \ell-i}} \prod_{w\in \text{Int}(T_{R_v})}\frac{\alpha_{I_{L_w},I_{R_w}}^w}{\alpha_{I_w}^w}\bigotimes_{u=a+1}^G\ket{0^{n_u-s_u}1^{s_u}}\right)}.
    \end{align*}
    Noting $\sum_{u=1}^a r_u +\sum_{u'=a+1}^G s_{u'} = \ell$, and $\text{Int}(T_{L_v})\cup \text{Int}(T_{R_v}) = \text{Int}(T_v)\setminus \{v\}
    $, this product can be expanded to have
    \begin{equation}\label{expanded induction hypothesis equation}
        \ket{\Psi} = \sum_{i=0}^\ell\beta_{i}^{v,\ell}\sum_{\substack{r_1,\dots, r_a\\ s_{a+1},\dots, s_G\\ \sum r_u = i\\
        \sum s_u = \ell-i}} \prod_{w\in\text{Int}(T)\setminus\{v\}} \frac{\alpha_{I_{L_w},I_{R_w}}^w}{\alpha_{I_w}^w}\bigotimes_{u=1}^a \ket{0^{n_u-r_u}1^{r_u}}\bigotimes_{u=a+1}^G\ket{0^{n_u-s_u}1^{s_u}}
    \end{equation}
    Since $\beta_{i}^{v,\ell} = \frac{\alpha_{i,\ell-i}^v}{\alpha_{\ell}^v}$ by definition and \Cref{expanded induction hypothesis equation} varies over all $i\leq \ell$, we can simplify it to arrive at 
    \begin{equation*}
        \ket{\psi} = \sum_{\substack{i_1,\dots, i_G\\ \sum i_u = \ell}}  \prod_{w\in\text{Int}(T)} \frac{\alpha_{I_{L_w},I_{R_w}}^w}{\alpha_{I_w}^w}\bigotimes_{u=1}^G \ket{0^{n_u-i_u}1^{i_u}}
    \end{equation*}
    as desired. \qed

\section{Proof of \Cref{Algorithm}}\label{sec:AlgorithmProof}
Note that
\begin{equation}\label{eq:result}
    U\ket{\Psi} = \left(\bigotimes_{u=1}^G E^u\right) \ket{\Psi} = \sum_{\substack{i_1,\dots, i_L \\ \sum i_u=\ell}}c(i_1,\dots, i_G)\bigotimes_{u=1}^GE^u\ket{0^{n_u-i_u}1^{i_u}}
\end{equation}
 where for a fixed $u$,
\begin{equation*}
    E^u\ket{0^{n_u-i_u}1^{i_u}} = \sum_{g_u}\eta_{u, i_u}(g_{u})\ket{g_{u}}.
\end{equation*}
Focusing on the main subsystem of $n$ qubits, observe that $\bigotimes_{u=1}^G\ket{g_{u}}$ is an $n$ qubit basis state with Hamming weight $\ell$ since $\sum_u i_u = \ell$. Now, Eq.~\eqref{eq:result} becomes
\begin{align*}
\sum_{\substack{i_1,\dots, i_G \\ \sum i_u=\ell}}c(i_1,\dots, i_G)\bigotimes_{u=1}^GE^u\ket{0^{n_u-i_u}1^{i_u}}=& \sum_{\substack{i_1,\dots, i_G \\ \sum i_u=\ell}}c(i_1,\dots, i_G)\bigotimes_{u=1}^G\left(\sum_{g_u}\eta_{u,i_u}(g_u)\ket{g_{u}}\right)\\
=& \sum_{\substack{i_1,\dots, i_G \\ \sum i_u=\ell}}c(i_1,\dots, i_G)\left(\sum_{g_1}\cdots\sum_{g_G}\bigotimes_{u=1}^G\eta_{u,i_u}(g_u)\ket{g_{u}}\right)\\
=& \sum_{\substack{i_1,\dots, i_G \\ \sum i_u=\ell}}c(i_1,\dots, i_G)\ket{\psi_I}\\
=&\sum_{b\in\mathcal{H}_{v,\ell}}\alpha_b\ket{b}.
\end{align*}
\qed

\renewcommand\refname{References Cited}
\bibliography{manuscript}
\bibliographystyle{IEEEtran}

\end{document}